\documentclass[letterpaper, 10 pt, conference]{ieeeconf}  
\usepackage{graphicx}
\usepackage{epstopdf}

\IEEEoverridecommandlockouts                              
\usepackage{amsmath, amssymb}
\usepackage{comment}
\DeclareMathOperator*{\Argmin}{Argmin}
\newtheorem{remarkEnv}{Remark}
\newtheorem{proposition}{Proposition}
\usepackage{url}
\usepackage{subcaption}
\newenvironment{remark}[1][]{\begin{remarkEnv}}{\hfill$\blacklozenge$\end{remarkEnv}}
\usepackage{epstopdf}
\usepackage{xcolor}

\title{\LARGE \bf
Vanishing Stacked-Residual PINN for State Reconstruction of Hyperbolic Systems
}

\author{Katayoun Eshkofti and Matthieu Barreau
\thanks{This work is partially supported by the Wallenberg AI, Autonomous Systems and Software Program (WASP) funded by the Knut and Alice Wallenberg Foundation and Digital Futures.}
\thanks{K. Eshkofti and M. Barreau are with the Division of Decision and Control Systems, Digital Futures, KTH Royal Institute of Technology, SE-100 44 Stockholm, Sweden {\tt\small \{eshkofti,barreau\}@kth.se}}%
}

\begin{document}

\maketitle
\thispagestyle{empty}
\pagestyle{empty}

\begin{abstract}
In a more connected world, modeling multi-agent systems with hyperbolic partial differential equations (PDEs) offers a compact, physics-consistent description of collective dynamics. However, classical control tools need adaptation for these complex systems. Physics-informed neural networks (PINNs) provide a powerful framework to fix this issue by inferring solutions to PDEs by embedding governing equations into the neural network. A major limitation of original PINNs is their inability to capture steep gradients and discontinuities in hyperbolic PDEs. To tackle this problem, we propose a stacked residual PINN method enhanced with a vanishing viscosity mechanism. Initially, a basic PINN with a small viscosity coefficient provides a stable, low-fidelity solution. Residual correction blocks with learnable scaling parameters then iteratively refine this solution, progressively decreasing the viscosity coefficient to transition from parabolic to hyperbolic PDEs. Applying this method to traffic state reconstruction improved results by an order of magnitude in relative $\mathcal{L}^2$ error, demonstrating its potential to accurately estimate solutions where original PINNs struggle with instability and low fidelity.
\end{abstract}

\section{INTRODUCTION}

Quasi-linear hyperbolic partial differential equations (PDEs) are crucial in modern control problems, emerging in a wide range of applications, from fluid dynamics to electrical energy transportation and traffic flow \cite{BastinCoron2016}. A notable example is the role of the Hamilton-Jacobi equation in optimal and predictive control, highlighting the ubiquity of hyperbolic PDEs across both physical modeling and abstract optimization applications in control theory. State reconstruction and identification of systems governed by hyperbolic PDEs is of fundamental interest, as it allows for estimating the complete evolution of the system from partial and noisy observations. This capability is key for monitoring distributed systems and enables subsequent control.

However, model identification and state reconstruction for quasi-linear hyperbolic PDEs are challenging due to their nonlinear dynamics, discontinuities, non-uniqueness, and infinite-dimensional nature. Traditional model-based approaches typically require precise knowledge of the system model, low dimensionality, and favorable theoretical conditions to guarantee convergence, which are often difficult to ensure for complex systems. On the other hand, commonly used machine learning methods \cite{POLSON_ML} require a large number of measurements, which often results in overfitting \cite{HuangLimitations}.

These issues encourage investigating learning-based approaches that handle model uncertainties and efficiently utilize data. To address this, the authors of \cite{raissi2019physics} introduced physics-informed neural networks (PINNs), integrating governing equations directly within neural networks. By embedding physical models into the loss function and penalizing deviations, PINNs effectively learn solutions from sparse and noisy data. Specifically, \cite{barreua_identification} demonstrated that PINNs simultaneously identify unknown model parameters, reconstruct traffic flow states from sparse vehicle data, and extend predictions. This reveals the potential of PINN-based approaches for state reconstruction in hyperbolic PDEs and serves as an incentive for the present research.

Since the original PINN development, recent modifications have improved performance in complex scenarios. For example, \cite{AmandaStacked} proposed a multi-fidelity stacking approach, iteratively training a PINN to refine outputs progressively. Physics-informed residual adaptive networks employ projected input coordinates within residual blocks featuring adaptive skip connections to address deep multilayer perceptron derivative initialization issues \cite{piratenet}. Integrating localized artificial viscosity into PINNs enables automatic learning of optimal viscosity, enhancing accuracy over non-adaptive methods \cite{COUTINHOadaptivevisco}. Additionally, sequential and hierarchical PINN structures, such as multi-stage neural networks, have shown unique capabilities \cite{MSNN}.

Although recent studies suggest PINNs outperform traditional deep learning methods \cite{Huang2020}, they have limitations, particularly for hyperbolic PDEs, where they struggle to capture sharp features like shocks or discontinuities \cite{HuangLimitations}. Addressing these challenges motivates the development of more effective state estimation methods for hyperbolic PDEs. Moreover, no evidence demonstrates that PINN effectively learns the hyperbolic PDE underlying traffic state models with acceptable accuracy \cite{HuangLimitations}, making this an open and intriguing research area in traffic control.

Our contribution is to address this issue by incorporating prior knowledge in the form of a function series through an effective combination of the vanishing viscosity method from applied mathematics, curriculum learning from machine learning, and stacked PINNs. This approach ensures convergence to the unique entropic hyperbolic solution.

This paper consists of five sections. Section 2 reviews hyperbolic PDE formulations and the role of vanishing viscosity in ensuring stable, unique solutions. Section 3 introduces the vanishing stacked-residual PINN methodology, outlining its architecture and training with decreasing viscosity. Section 4 presents numerical experiments on traffic state reconstruction, highlighting shock-capturing capabilities and comparing performance with the original PINN. Section 5 concludes with key findings and future research directions.

\textbf{Notation:} Let $\mathbb{R}$ denote real numbers and $\mathbb{R}^{+}$ nonnegative real numbers. For differentiable single-variable functions, the prime notation $f'$ represents derivatives. Multivariate functions' partial derivatives with respect to space and time are denoted by $\partial f$ with corresponding subscripts. Moreover, $C^{1}(\mathbb{R}^{+} \times \mathbb{R})$ denotes continuously differentiable functions on the given domain. The sets $L^p(\Lambda)$, $H^k(\Lambda)$, and $L^{\infty}(\Lambda)$ represent Lebesgue, Sobolev, and essentially bounded function spaces on $\Lambda$, respectively. 

\section{Background on problem statement}

The main goal of this paper is to find a solution $u$ to the following 1d quasi-linear hyperbolic PDE posed on the spatiotemporal domain $\Lambda = [0,T] \times [0,L] \subset \mathbb{R}^+ \times \mathbb{R}$:
\begin{equation} \label{eq:system}
\begin{cases}
    \partial_t u + \partial_x f(u) = 0, & (t,x) \in \Lambda, \\
    u(0,x) = u^0(x), & x \in [0,L], \\
    u(t,0) = u_b^{-}(t), \ u(t,L) = u_b^{+}(t), & t \in [0,T].
\end{cases}
\end{equation}
Here, $u^0$ is a function of bounded variation that represents the initial data \cite[Definition~1.7.1]{Dafermos:1315649}. Suitable boundary conditions $u_b^{\pm}$ should be considered for the PDE in \eqref{eq:system} to be well-posed \cite{BastinCoron2016}. It is assumed that $f$ is a smooth flux function, at least $f \in C^2(\mathbb{R})$, so that the PDE is strictly hyperbolic. Moreover, to ensure that the mapping $u \mapsto f(u)$ remains bounded, it is assumed that $f$ is globally Lipschitz or has at most polynomial growth. 

Since hyperbolic solutions can develop discontinuities in finite time, we can only investigate weak solutions, which belong to a distribution space \cite{BastinCoron2016}.
However, these weak solutions are usually not unique unless an entropy condition is imposed. 
To this end, one can consider the Lax-E entropy condition, which adds a point-wise constraint in the form of $f'(u(t,x^-)) \leq f'(u(t,x^{+}))$ for $t,x \in \Lambda$. A weak solution that satisfies the entropy inequality is then called entropy-admissible. As stated in \cite[Theorem~14.10.2]{Dafermos:1315649}, there is a unique entropic solution to \eqref{eq:system}.

Testing the Lax-E inequality at each point is numerically intractable. A more robust and practical way to construct weak entropy solutions is through the vanishing viscosity method. Consider the parabolic regularization of \eqref{eq:system} as
\begin{equation}\label{eq:parabolic}
    \partial_t u_{\gamma} + \partial_x f(u_{\gamma}) = \gamma \partial_{xx} u_{\gamma}, \quad \gamma > 0.
\end{equation}
where $\gamma$ is a small viscosity coefficient. It is worth noting that $u_{\gamma}$ signifies the solution corresponding to viscosity $\gamma$. Under the standard assumptions previously mentioned, the parabolic PDE in \eqref{eq:parabolic} results in a unique classical solution $u_{\gamma}$ for each fixed $\gamma$ \cite[Theorem~14.6]{amann1993nonhomogeneous}. By finding estimates that are uniform with respect to $\gamma$, it is proved that the sequence $\{u_{\gamma}\}$ remains bounded in appropriate norms \cite{dafermos2016hyperbolic}. 
Uniform energy estimates and the derivation of an entropy inequality enable passing to the limit as 
$\gamma \to 0$, thereby, as discussed in Chapter 2 of \cite{dafermos2016hyperbolic}, obtaining a weak solution that is entropy-admissible and satisfies the additional regularity conditions outlined as $ u \in C^0\left([0,T]; H^2([0,L])\right) \cap C^1\left([0,T]; H^1([0,L])\right)$.

The rigorous justification of this limit passage and the well-posedness of the corresponding problem is established via a combination of compactness arguments and the construction of a basic quadratic Lyapunov function, as detailed in \cite{BastinCoron2016}. Moreover, the dissipativity of the boundary conditions is crucial to ensuring that the energy associated with the system decays over time. Consequently, by combining the vanishing viscosity method \cite{Kružkov_1970} with uniform energy and entropy estimates, the existence, uniqueness, and entropy admissibility of the solution to \eqref{eq:system} are guaranteed.

Although vanilla PINNs have been successfully applied to various types of problems, particularly those classified as parabolic PDEs, a major limitation is their poor performance in terms of convergence and accuracy when solving hyperbolic PDEs \cite{HuangLimitations}. 
Therefore, a modification to the PINN structure is necessary to enhance its ability to solve hyperbolic PDEs, leading to the following problem.

\textbf{Problem:} We want to efficiently and accurately approximate the entropic solution $u$ of a hyperbolic PDE specified in \eqref{eq:system} using PINNs. 

To achieve this objective, we introduce a novel variant of PINN called vanishing stacked-residual PINN, which incrementally refines a baseline approximation through stacked residual-correction subnetworks with vanishing viscosity. This approach incorporates both the PDE residual and data to enforce the governing equations while also guiding the solution in regions where the PDE may be insufficient or highly nonlinear.

In the following section, we demonstrate how adopting this smooth-to-sharp transition in the stacked residual PINN enhances the capture of discontinuities, enabling high-resolution reconstruction of hyperbolic PDE states.

\section{Methodology}

This section first presents an overview of the proposed architecture, followed by a detailed discussion of its application to hyperbolic PDEs.

\subsection{Classical PINN}

In the standard formulation of PINNs, the goal is to approximate the solution $u$ of PDE  \eqref{eq:system} using a dense feedforward neural network $\hat{u}(\cdot; \boldsymbol{\theta})$. 
Let the neural network have $L$ hidden layers, formulated as follows:
\begin{equation}
    \begin{array}{rl}
        \hat{u}(t,x; \boldsymbol{\theta}) \!\!\!\!&= W_L \times H_{L-1} \circ \cdots \circ H_1(t,x) + b_L \\
        &\triangleq \mathcal{N}([t, x]; \boldsymbol{\theta}),
    \end{array}
\end{equation}
for $(t,x) \in \Lambda$ where $(t, x)$ represents the input coordinates in the spatiotemporal domain. For each hidden layer $l = 1, \dots, L-1$, the feature map is defined as $H_l(\nu) = \phi.(W_l \nu + b_l)$.
The weights and biases at layer $l$ are represented by $W_l \in \mathbb{R}^{n_{l} \times n_{l-1}}$ and $b_l \in \mathbb{R}^{n_l}$. The entire set of network parameters is denoted $\boldsymbol{\theta} = \{ W_l, b_l \}_{l=1}^{L}$ and the element-wise activation function is $\phi \in C^{\infty}(\mathbb{R}, \mathbb{R})$.

We are interested in approximating the  unique entropic solution $u$ to \eqref{eq:system}, meaning that we want to solve the following constrained optimization problem:
\[
    \boldsymbol{\theta}^* = \begin{array}[t]{cl}
        \Argmin_{\boldsymbol{\theta}} & \displaystyle \int_{\Gamma} \| u(\nu) - \hat{u}(\nu) \|^2 d\nu\\
        \text{s. t.} & \displaystyle \int_{\Lambda} \left| r_0(\cdot; \hat{u}(\cdot; \boldsymbol{\theta})) \right|^2 = 0,
    \end{array}
\]
where $\Gamma = \Gamma_{init} \cup \Gamma_{boundary}$ is the measured boundary of System~\eqref{eq:system} with
\[
    \begin{array}{l}
        \Gamma_{init} = \left\{ (0, x) \ | \ x \in [0, L] \right\}, \\
        \Gamma_{boundary} = \left\{ (t, 0) \ | \ t \in [0, T] \right\} \cup \left\{ (t, L) \ | \ t \in [0, T] \right\},
    \end{array}
\]
and $\hat{u}$ refers to $\hat{u}(\cdot, \boldsymbol{\theta})$ to ease the reading. The residual $r$ is also defined as:
\begin{equation}
    \label{eq:residual}
    r_{\gamma}(\cdot; \hat{u}) = \partial_t \hat{u} + \partial_x f(\hat{u}) - \gamma \partial_{xx} \hat{u}.
\end{equation}

The previous problem cannot be numerically solved because it contains integrals. Following the methodology in \cite{barreau_accuracy} using Monte-Carlo sampling and the Lagrangian formulation, we get the following relaxed but numerically tractable optimization problem:
\[
    \boldsymbol{\theta}^* = \Argmin_{\boldsymbol{\theta}} \ \max_{\lambda > 0} \ \mathcal{L}_{\lambda}(\hat{u}, 0),
\]
where $\mathcal{L}_{\lambda}(\hat{u}, \gamma) = \mathcal{L}_{data}(\hat{u}(\cdot, \theta)) + \lambda \mathcal{L}_{phy}(\hat{u}(\cdot, \theta), \gamma)$.
\[
    \mathcal{L}_{data}(\hat{u}) = \frac{1}{|\mathcal{D}_{data}|} 
    \sum_{(t,x) \in \mathcal{D}_{data}} \left| u(t, x) - \hat{u}(t, x) \right|^2,
\]
\[
    \mathcal{L}_{phy}(\hat{u}, \gamma) = \frac{1}{|\mathcal{D}_{phy}|} \sum_{(t,x) \in \mathcal{D}_{phy}} \left| r_{\gamma}(t, x; \hat{u}) \right|^2,
\]
with $\mathcal{D}_{phy} \subset \Lambda$ a discrete set of cardinal $|\mathcal{D}_{phy}|$ and $\mathcal{D}_{data} \subset \Gamma$ a discrete set of cardinal $|\mathcal{D}_{data}|$. This problem is typically solved using a primal-dual strategy, alternating between gradient descent on the $\min$ problem over the tensor of parameters $\theta$ and the $\max$ over the Lagrange multiplier $\lambda$. A simpler and often efficient method is to fix $\lambda$ to a small constant and apply gradient descent only to the $\min$ problem. For simplicity, we consider this option in this paper.

\begin{figure}[!t]
    \centering
    \vspace{-1mm}    \includegraphics[width=0.2\textwidth]{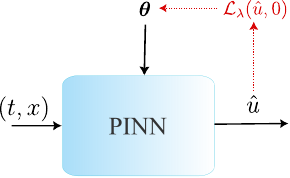}
    \caption{Vanilla PINN structure. Black lines are forward pass, and red-dotted lines are backpropagation.}
    \label{fig:classic}
\end{figure}

The total loss function consists then of two terms: the data mismatch $\mathcal{L}_{data}$ and the PDE residual $\mathcal{L}_{phy}$. These losses are computed using the mean square error between the true and the approximated solution, transforming the optimization problem into a learning one. For the physics loss, we want to minimize the PDE residual $r_{\gamma}(\cdot; \boldsymbol{\theta})$ over $\Lambda$, bringing regularization and forcing the convergence to the solution of the PDE. The general structure of PINN is illustrated in Fig.~\ref{fig:classic}. 

As explained in \cite{barreau_accuracy, sirignano2018dgm}, the PINN methodology only succeeds in the case of a Lipschitz continuous PDE operator, which is never the case in hyperbolic problems such as \eqref{eq:system}. The following section provides a detailed discussion of a refined architecture that proposes a solution.

\subsection{Residual PINN}

First, consider solving the regularized PDE in \eqref{eq:parabolic}. If $\gamma$ is small enough, then the approximated solution will be close to the entropic solution of \eqref{eq:system}, based on the vanishing viscosity method \cite{Kružkov_1970}. To this end, the initial stage consists of a baseline PINN designed to approximate the solution $u_{\gamma_{\text{init}}}$ to the parabolic PDE \eqref{eq:parabolic} by $\hat{u}^{(0)}(\cdot; \boldsymbol{\theta}_0)$. At this stage, the viscosity coefficient $\gamma_{\text{init}}$ is chosen to be large enough to ensure that $\hat{u}$ is entropic and that the PDE operator is sufficiently Lipschitz to guarantee that the PINN can successfully learn the approximated solution. 

One key idea of this article is to combine the vanishing methodology with the residual-network architecture \cite{he2016deep}. Subsequently, we introduce a residual stage where $\hat{u}^{(0)}$ is refined by a single residual correction network to get $\hat{u}^{(1)}$, as shown in Fig.~\ref{fig:residual}. This residual-correction PINN bears a close resemblance to a Luenberger observer that feedbacks the difference $u - \hat{u}$ to update its estimate. In both cases, the discrepancy triggers the correction.
At this second stage, the viscosity coefficient is $\gamma_1 = 0$, hence the solution to the hyperbolic PDE in \eqref{eq:system} is approximated by $\hat{u}^{(1)}$. Consequently, the approximation at the second stage is: 
\begin{equation}
    \hat{u}^{(1)}(t, x; \boldsymbol{\theta}_1) = \hat{u}^{(0)}(t, x; \boldsymbol{\theta}_{0})  + \lvert \alpha_1 \rvert\mathcal{N} \big( [t, x, \hat{u}^{(0)}(t,x)]; \boldsymbol{\theta}_1 \big).
\label{eq:rescorrection}
\end{equation}
In practice, the residual block is a feedforward neural network whose inputs include the spatiotemporal coordinates $(t, x)$ and the approximation at the previous step $\hat{u}^{(0)}(t,x)$. The scaling factor $\lvert \alpha_1 \rvert$ is a learnable variable that determines and controls the extent to which each residual block contributes to the update. The coefficient $\alpha_1$ is introduced because we want to keep the correction small, such that we only correct slightly the parabolic solution, ensuring that the hyperbolic solution remains entropic. Indeed, a large $\alpha_1$ will probably signify that $\hat{u}^{(1)}$ is quite different from $\hat{u}^{(0)}$ which does not align with the vanishing viscosity method. The new optimization problem is
\[
    \hspace*{-0.2cm}\boldsymbol{\theta}_0^*, \boldsymbol{\theta}_1^*  = \Argmin_{\boldsymbol{\theta}_0,\boldsymbol{\theta}_1,\alpha_1} \ \mathcal{L}(\boldsymbol{\theta}_0, \boldsymbol{\theta}_1) + \alpha_1^2,
\]
where
\begin{equation*}
    \mathcal{L}(\boldsymbol{\theta}_0, \boldsymbol{\theta}_1) = \frac{1}{2} \left( \mathcal{L}_{\lambda}(\hat{u}^{(0)}, \gamma_{\text{init}}) + \mathcal{L}_{\lambda}(\hat{u}^{(1)}, 0) \right).
\end{equation*}

\begin{figure}[!t]
    \centering
    \vspace{-1mm}    \includegraphics[width=0.33\textwidth]{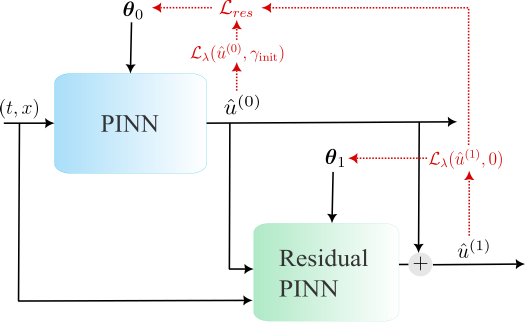}
    \caption{Residual PINN structure. Black lines are forward pass, and red-dotted lines are backpropagation.}
    \label{fig:residual}
\end{figure}

As discussed in Section 2, the vanishing viscosity principle ensures that the solution to the parabolic PDE in \eqref{eq:parabolic}, with a small viscosity term, converges to the entropic weak solution to the corresponding hyperbolic PDE as the viscosity coefficient decreases to zero. 

As a result, training a PINN on the parabolic form in the initial stage, which admits a unique smooth solution, and then transitioning to the hyperbolic PDE allows for acquiring both the stability of the parabolic regime and the fidelity of the hyperbolic limit.

\subsection{Stacked residual PINN and vanishing viscosity approach}

While a single residual correction block may succeed in simple problems, it can fail when dealing with sharp and localized solution features. Moreover, relying on a single residual PINN to approximate both smooth regions and shocks imposes a burden on the learning process. 

The other key idea of this article is to adapt the iterative stacking method put forward by \cite{AmandaStacked}. In our study, by stacking multiple residual blocks, each trained with a successively smaller $\gamma$, a multi-stage correction process is implemented following the vanishing viscosity method. In fact, each stage leverages the well-posedness of the parabolic PDE in \eqref{eq:parabolic} to compute a smooth and stable correction, which is then gradually refined to capture the sharper features of the hyperbolic PDEs. In other words, $\hat{u}^{(0)}$ is incrementally refined by stacking $n$ residual-correction subnetworks. The approximation 
at the $i$th stage is given by:
\begin{multline} \label{eq:stackedrescorrection}
    \hat{u}^{(i)}(t, x; \boldsymbol{\theta}_i) = 
    \hat{u}^{(i-1)}(t, x; \boldsymbol{\theta}_{i-1}) \\
    + \lvert \alpha_i \rvert \mathcal{N} \big( [t, x, \hat{u}^{(i-1)}(t,x)]; \boldsymbol{\theta}_i \big),
\end{multline}
where $i \in \{1, ..., n\}$, and $\lvert \alpha_i \rvert$ is an adaptive parameter that scales the contribution of each residual block. As in the viscous Burgers equation, the Cole–Hopf transformation \cite{cole1951quasi,hopf1950partial} demonstrates that solutions depend exponentially on the viscosity parameter. As $\gamma$ decreases, the transformed solution exhibits rapid exponential changes. Therefore, integrating a superlinear $\gamma$, as 
\begin{equation} \label{eq:vanishingdiffuison}
    \gamma_i = \gamma_{\text{init}} \left[ 1 - \left( \frac{i}{n} \right)^p \right]
\end{equation}
with $p > 1$ so that $\gamma_0 = \gamma_{\text{init}} > 0$ and $\gamma_n = 0$, ensures that the stacked residual blocks smoothly capture the transition from parabolic to hyperbolic behavior without compromising stability. 

As in stage 1, a physics-informed loss is imposed at each level $i$, ensuring that $r_{\gamma_i}(\cdot; \boldsymbol{\theta}_i)$ remains small along with the usual boundary and initial conditions.

In other words, at stacked residual block $i$, the PDE residual is governed by $r_{\gamma_i}$ such that the first block gets the viscous coefficient $\gamma_{\text{init}}$ and the last stage $\gamma_n = 0$. The optimization problem can be formulated as:
\[
    \left\{ \boldsymbol{\theta}_i^*,\alpha_i^* \right\}_{i=0}^{n} = 
    \Argmin_{\{\boldsymbol{\theta}_i, \alpha_i\}_{i=0}^{n}} \ \mathcal{L} \left( \left\{ \boldsymbol{\theta}_i \right\}_{i=0}^{n} \right) 
    + \sum_{i=1}^{n} \alpha_i^2
\]
where
\begin{equation*}
    \mathcal{L} \left( \left\{ \boldsymbol{\theta}_i \right\}_{i=0}^{n} \right) =
    \frac{1}{n+1} \sum_{i=0}^{n} \mathcal{L}_{\lambda}(\hat{u}^{(i)}, \gamma_i).
\end{equation*}
Here $\mathcal{L}$ represents the total loss function, which consists of the base PINN and the stacked residual PINNs. The residual PINN block diagram is presented in Fig.~\ref{fig:stacked}.

\begin{figure}[!t]
    \centering
    \vspace{-1mm}    \includegraphics[width=\columnwidth]{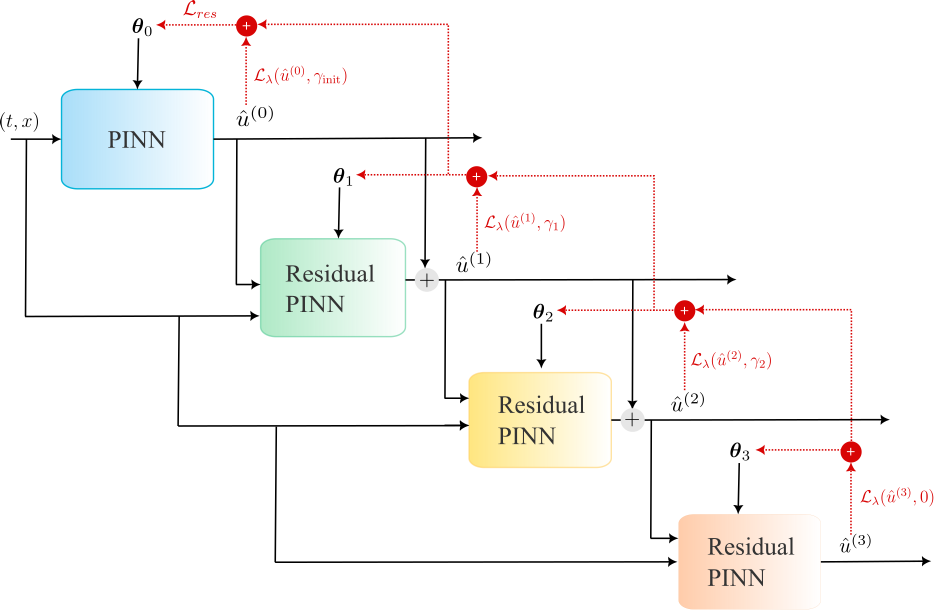}
    \caption{Vanishing stacked-residual PINN. Black lines are forward pass, and red-dotted lines are backpropagation.}
    \label{fig:stacked}
\end{figure}

\begin{remark}
The stacked PINN method was initially introduced in \cite{AmandaStacked}, where they employed a multi-fidelity strategy that stacks PINNs. In this approach, each network’s output provides a lower-fidelity input for the next stage, incrementally refining the model’s expressivity. We share the same idea; however, our method employs a residual PINN at each stage. Contrary to \cite{AmandaStacked}, we use a different residual $r_{\gamma}$ at each stage to align with the vanishing viscosity method.
\end{remark}

With this approach, we obtain certain convergence guarantees, as stated in the following proposition.
\begin{proposition}
Consider the hyperbolic PDE in \eqref{eq:system}. Let \( \gamma_{\text{init}} > 0 \) and consider a sequence of viscosity coefficients defined as in \eqref{eq:vanishingdiffuison}, ensuring \( \gamma_n = 0 \). Assume the following:

\begin{enumerate}
    \item For each \( \gamma_i > 0 \), the parabolic PDE in \eqref{eq:parabolic} admits a unique smooth solution \( u_{\gamma_i} \), converging strongly in \(L^1(\Lambda) \) to the entropy solution of \( u \) in \eqref{eq:system} as \( \gamma_i \to 0 \).

    \item At each stage \( i \), the neural network \( \hat{u}^{(i)} (\cdot,\boldsymbol{\theta}_i) \) has sufficient expressivity such that the approximation error is bounded by \( \varepsilon_i \), that is:
    \[
        \|\hat{u}^{(i)} - u_{\gamma_i}\|_{L^2(\Lambda)} \leq \varepsilon_i, \quad \text{where } \varepsilon_i \to 0 \text{ as } i \to n.
    \]
\end{enumerate}

Then, the stacked residual PINN solution \( \hat{u}^{(n)} \) converges to the entropy solution of \eqref{eq:system}. 
\end{proposition}

\begin{proof} 
According to the vanishing viscosity theorem \cite{Kružkov_1970}, for the sequence \( \{\gamma_i\} \) with \( \gamma_i \to 0 \) as \( i \to n \), we have $ \lim_{i \to n} \| u_{\gamma_i} - u \|_{{L}^1(\Lambda)} = 0$.

Considering the second assumption \cite{hornik1991approximation}, the residual correction mechanism ensures that the discrepancy between \( \hat{u}^{(i)} \) and the viscous solution \( u_{\gamma_{i}} \) is minimized through the physics loss \( \mathcal{L}_{\text{phy}} (\gamma_i, \hat{u}^{(i)}) \). The total error at stage \( i \) decomposes as:
\[
    \left\| \hat{u}^{(i)} - u \right\|_{{L}^1(\Lambda)} \leq 
    \underbrace{
    \sqrt{|\Lambda|}
    \left\| \hat{u}^{(i)} - u_{\gamma_i} \right\|_{{L}^2(\Lambda)}}_{\varepsilon_i} 
    + 
    \underbrace{\left\| u_{\gamma_i} - u \right\|_{{L}^1(\Lambda)}}_{\delta_i}.
\]
By the first assumption, \( \delta_i \to 0 \) as \( \gamma_i \to 0 \). By the second assumption, with a sufficiently large network, \( \varepsilon_i \to 0 \).
For the sequence \( \{ \hat{u}^{(i)} \} \), we have
\[
    \| \hat{u}^{(n)} - u \|_{{L}^1(\Lambda)} \leq \sqrt{|\Lambda|}\sum_{i=0}^{n} \varepsilon_i + \sum_{i=0}^{n} \delta_i.
\]
Since both series converge under assumptions 1 and 2, the proof is complete.
\end{proof}

\section{Results}

In this section, our goal is to estimate vehicle density over a road $[0,L]$ from local density measurements by employing the proposed vanishing stacked-residual PINN \footnote{The code and data are available at \url{https://github.com/KatayounEshkofti/VanishingStackedResidualPINN}}.

In the example considered, the PDE \eqref{eq:system} represents the LWR model \cite{RichardShock,lighthill1955kinematic}. The normalized density $u$ is defined such that \( u = 0 \) corresponds to an empty road, while \( u = 1 \) represents bumper-to-bumper traffic conditions. Two density measurements are provided at the boundaries \( x \in \{0, L\} \) and can be collected using loop detectors.

Additionally, \( f(u) = V_f u(1 - u) \) is referred to as the Greenshields flow equation \cite{greenshields1935study}, which describes the relationship between vehicle velocity and density. The parameter \( V_f \) denotes the free-flow velocity. The Greenshields function serves as the advection term, and its nonlinear dependence on \( u \) can lead to shock formation or discontinuities.

By employing the vanishing viscosity limit, the governing equation in \eqref{eq:parabolic} is considered, and a small initial viscosity coefficient \( \gamma_{\text{init}} = 0.1 \). To investigate the impact of stacked layers on the improvement of results, the algorithm is implemented while considering four different numbers of stacked residual blocks, $n = \{0, 1, 3, 5\}$, where $n = 0$ represents vanilla PINN. To compute the data loss $\mathcal{L}_{\text{data}}$ in the total loss function of PINN, density measurements are supplied via Godunov simulation of \eqref{eq:system}. 

The baseline PINN at the initial stage consists of three hidden layers, each with 30 neurons. Each residual correction block comprises three hidden layers, each with 40 neurons. Vanishing viscosity is implemented via the function defined in \eqref{eq:vanishingdiffuison}, where \( p = 2 \).

It is worth noting that, to ensure a fair comparison across all scenarios, the same number of hidden layers and neurons is utilized. Additionally, a $15{,}000$-iteration schedule is set for all scenarios, combined with an early stopping criterion using a patience parameter, so that training would automatically terminate once the loss plateaus. Under these conditions, the modified stacked PINN converged and stopped at iteration 11,000, while the other variants continued running through all 15,000 iterations.

The results are compared with solutions obtained from the Godunov simulation, and the relative \( \mathcal{L}^2 \) error is defined as the evaluation metric:
\begin{equation}
    \vspace*{-0.2cm}
    \text{relative } \mathcal{L}^2 = \frac{\|u - \hat{u}\|_2}{\|u\|_2}.
\label{eq:error}
\end{equation}
The relative \( \mathcal{L}^2 \) error quantifies the percentage discrepancy between the estimated density and exact values. As shown in Table~\ref{table:comparison}, increasing stacked residual blocks from 0 to 5 consistently reduces the error from about 19\% to 4\%. This suggests that additional residual blocks enhance estimation accuracy, though the smaller gain from 3 to 5 blocks may indicate diminishing returns.

\begin{table} 
    \centering
    \renewcommand{\arraystretch}{1.2}
    \begin{tabular}{c||c|c|c|c}
        \# Stacked residual blocks & 0 & 1 & 3 & 5 \\
        \hline
        Relative $\mathcal{L}^2$ error ($\times 10^{-2}$) & 18.98 & 13.86 & 4.86 & 4.66 \\
    \end{tabular}
    \caption{Relative $\mathcal{L}^2$ error for various numbers of stacked residual blocks.}
    \label{table:comparison}
\end{table}

\begin{figure}[!t]
    \centering
    \includegraphics[width=\columnwidth]{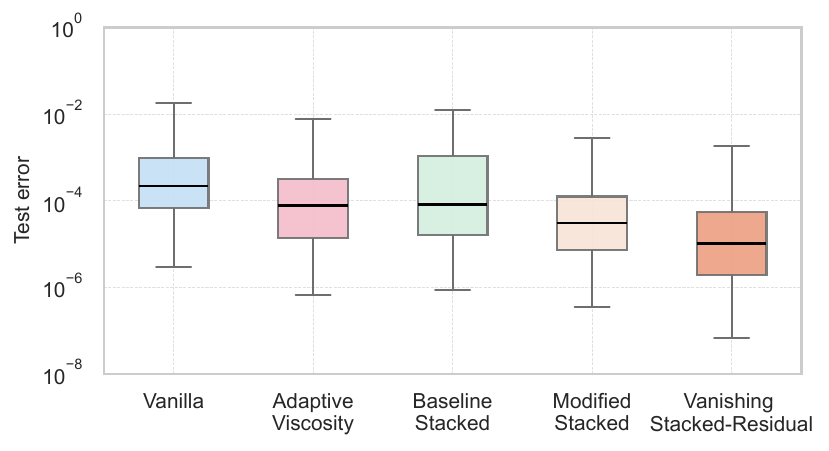}
    \caption{Distribution of errors over the test data for five methods applied to the traffic reconstruction problem.}
    \label{fig:boxplot}
\end{figure}

Moreover, we plot the distribution of point-wise errors across the entire spatio-temporal test domain in Fig.~\ref{fig:boxplot} to compare the performance of our method with the vanilla PINN \cite{raissi2019physics}, PINN with adaptive localized artificial viscosity \cite{COUTINHOadaptivevisco}, and the original stacked PINN \cite{AmandaStacked}. Note that the modified stacked PINN \cite{AmandaStacked} integrates vanishing viscosity into stacked networks. The baseline stacked PINN, the modified stacked PINN, and the vanishing stacked-residual PINN each consist of three stacked networks, with three hidden layers and 40 neurons per network. For fairness, the vanilla PINN and PINN with adaptive artificial viscosity consist of nine hidden layers, each with 40 neurons. As shown in Fig.~\ref{fig:boxplot}, although multiple shocks are present, the proposed vanishing stacked-residual PINN achieves a lower median error over test data. The accuracy of the proposed method can be seen from Fig.~\ref{fig:numerical_simulation}, where both the speed and shape of the propagating shocks are correctly caught. More specifically, Fig.~\ref{fig:residual_block} shows that the first residual block $\alpha_1 \mathcal{N}(\cdot, \theta_1)$ is far from being zero, indicating that it is grasping more shocks and sharpening the existing ones from the previous block.

These results indicate that progressively stacking networks enhances reconstruction accuracy, while residual correction networks further stabilize learning, leading to lower errors and reduced variability in state estimations.

\begin{figure}[!t]
     \centering
     \begin{subfigure}[b]{0.5\textwidth}
         \centering
         \includegraphics[width=0.9\textwidth]{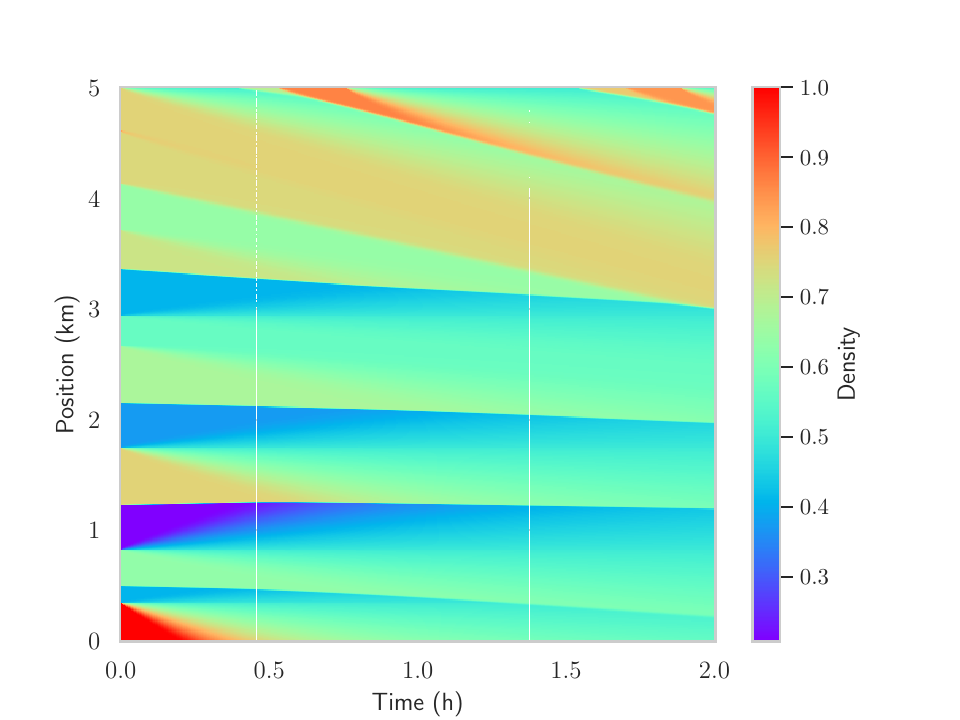}
         \caption{$u$}
         \label{fig:y equals x}
     \end{subfigure}
     \begin{subfigure}[b]{0.5\textwidth}
         \centering
         \includegraphics[width=0.9\textwidth]{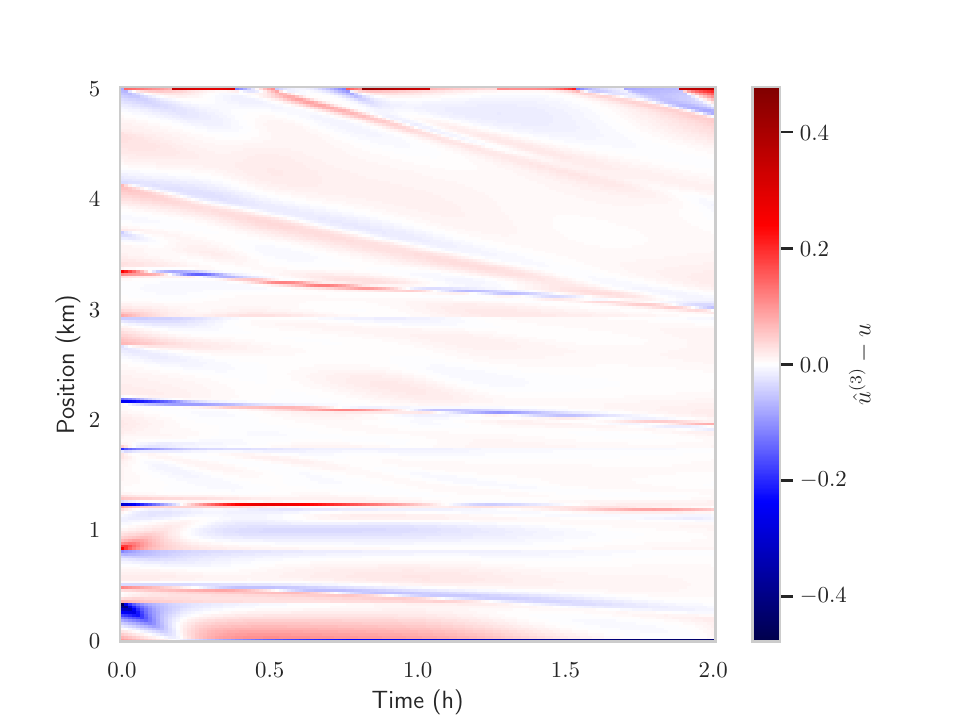}
         \caption{$\hat{u}^{(3)}-u$}
         \label{fig:three sin x}
     \end{subfigure}
    \caption{Numerical examples on the state reconstruction.}
    \label{fig:numerical_simulation}
\end{figure}

\begin{figure}[!t]
    \centering
    \vspace{-1mm}    \includegraphics[width=0.45\textwidth]{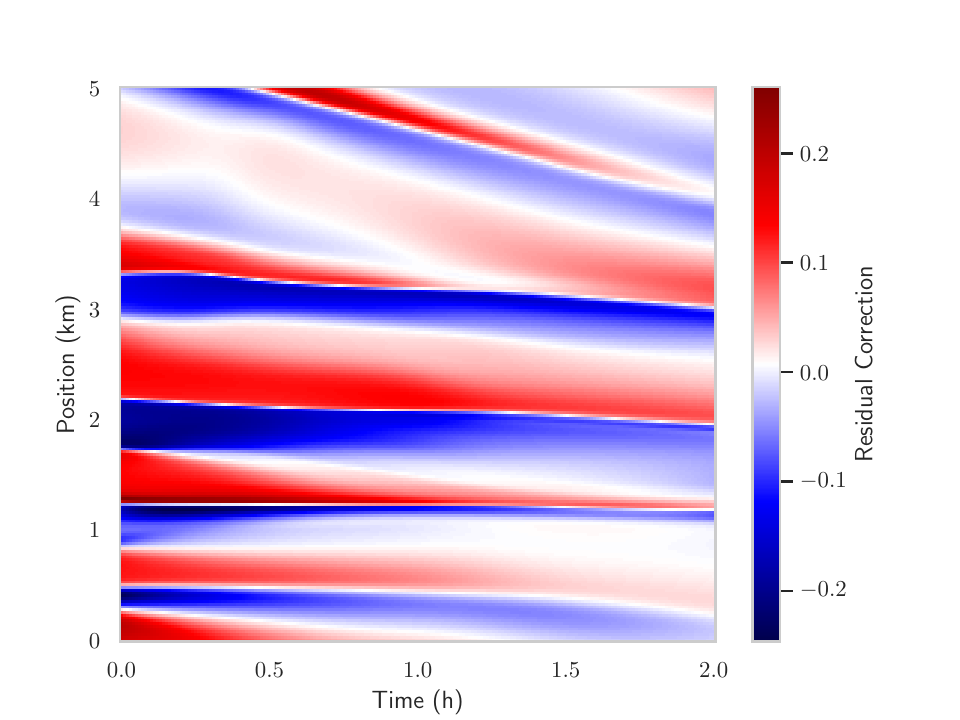}
    \caption{Value of the residual block $\alpha_1 \mathcal{N}(\cdot; \theta_1)$.}
    \label{fig:residual_block}
\end{figure}


\section{Conclusion and future works}

This paper introduced the vanishing stacked-residual PINN, a variation of the PINN framework designed to reconstruct states of hyperbolic PDEs. By adopting the vanishing viscosity principle, the approach progressively enhances solution approximation, effectively capturing discontinuities and shocks. Applied to a traffic state reconstruction problem, the method demonstrated an order of magnitude improvement in accuracy over the vanilla PINN.

Future work will focus on extending the methodology to higher-dimensional and more complex hyperbolic PDEs, aiming to confirm its broader applicability, particularly in control engineering.


\bibliographystyle{IEEEtran}
\bibliography{biblio}

\end{document}